\documentclass[12pt,reqno]{amsart}

\newcommand\version{January 10, 2022}

\usepackage{amsmath,amsfonts,amsthm,amssymb,amsxtra}
\usepackage{bbm} 
\usepackage{xcolor}

\newtheorem{theorem}{Theorem}
\newtheorem{proposition}[theorem]{Proposition}
\newtheorem{lemma}[theorem]{Lemma}

\theoremstyle{definition}

\theoremstyle{remark}

\newtheorem{remark}[theorem]{Remark}

\newcommand{\1}{\mathbbm{1}}

\newcommand{\C}{\mathbb{C}}

\newcommand{\const}{\mathrm{const}\ }

\renewcommand{\epsilon}{\varepsilon}

\newcommand{\loc}{{\rm loc}}

\renewcommand{\phi}{\varphi}
\newcommand{\R}{\mathbb{R}}

\newcommand{\Sph}{\mathbb{S}}

\DeclareMathOperator{\im}{Im}

\DeclareMathOperator{\re}{Re}

\newcommand{\ab}{|\psi|_{\varepsilon}}

\begin{document}

\title[A sharp criterion for zero modes of the Dirac equation --- \version]{A sharp criterion for zero modes\\ of the Dirac equation}

\author{Rupert L. Frank}
\address[Rupert L. Frank]{Mathe\-matisches Institut, Ludwig-Maximilans Universit\"at M\"unchen, The\-resienstr.~39, 80333 M\"unchen, Germany, and Munich Center for Quantum Science and Technology, Schel\-ling\-str.~4, 80799 M\"unchen, Germany, and Mathematics 253-37, Caltech, Pasa\-de\-na, CA 91125, USA}
\email{r.frank@lmu.de}

\author{Michael Loss}
\address[Michael Loss]{School of Mathematics, Georgia Institute of
	Technology, Atlanta, GA 30332-0160, USA}
\email{loss@math.gatech.edu}

\renewcommand{\thefootnote}{${}$} \footnotetext{\copyright\, 2022 by the authors. This paper may be reproduced, in its entirety, for non-commercial purposes.\\
Partial support through US National Science Foundation grants DMS-1954995 (R.L.F.) and DMS-1856645 (M.L.), as well as through the Deutsche Forschungsgemeinschaft (DFG, German Research Foundation) through Germany’s Excellence Strategy EXC-2111-390814868 (R.L.F.) is acknowledged.}

\begin{abstract}
	It is shown that $\Vert A \Vert_{L^d}^2 \ge \frac{d}{d-2}\, S_d$ is a necessary 
	condition for the existence of a nontrivial solution of the Dirac equation $\gamma \cdot (-i\nabla -A)\psi = 0$ in $d$ dimensions.
	Here, $S_d$ is the sharp Sobolev constant. If $d$ is odd and $\Vert A \Vert_{L^d}^2= \frac{d}{d-2}\, S_d$, then there exist vector
	potentials that allow for zero modes. A complete classification of these vector potentials and their corresponding zero modes is given.
\end{abstract}

\maketitle

\section{Introduction and main result}

In this paper we are interested in sharp nonexistence results for nontrivial solutions of the zero mode equation
\begin{equation}
	\label{eq:eq}
	\gamma\cdot(-i\nabla-A)\psi = 0 
	\qquad\text{in}\ \R^d \,.
\end{equation}
It can be considered a sequel to our previous work  \cite{FrLo1}, to which we refer the reader for more background and references. Throughout, we will be working in spatial dimensions $d\geq 3$. Let
$$
N := 2^{[d/2]} \,.
$$
In \eqref{eq:eq}, $\gamma_1,\ldots,\gamma_d$ are Hermitian $N\times N$ matrices satisfying
$$
\gamma_j \gamma_k + \gamma_k \gamma_j = 2 \delta_{j,k}
\qquad\text{for all}\ 1\leq j,k\leq d \,.
$$
Moreover, for a vector $a\in\R^d$ we set $\gamma\cdot a := \sum_{j=1}^d \gamma_j a_j$. The gamma matrices are the generalization to higher dimensions of the usual Pauli matrices and reduce to them in dimension $d=3$. It is known that the gamma matrices are unique up to a simultaneous unitary conjugation.

The quantity $A$ in \eqref{eq:eq} is a vector field on $\R^d$. We will assume throughout that
\begin{equation}
	\label{eq:aass}
	A\in L^d(\R^d,\R^d) \,.
\end{equation}
The $L^d$ norm of $A$ appears naturally in this problem, as we will see below. Physically, $A$ is the vector potential corresponding to the magnetic field $\nabla\wedge A$. (This magnetic field is, in general, only defined as a distribution.)

Finally, the quantity $\psi$ in \eqref{eq:eq} is a spinor field, that is, a function from $\R^d$ to $\C^N$. We will assume that
$$
\psi \in L^p(\R^d,\C^N)
$$
for some $d/(d-1)<p<\infty$. We have shown in \cite{FrLo1} that, under assumption \eqref{eq:aass}, if $\psi\in L^p$ for \emph{some} $d/(d-1)<p<\infty$, then $\psi\in L^p$ for \emph{all} $d/(d-1)<p<\infty$.

We emphasize that we do not require any further assumptions besides \eqref{eq:aass} and $\psi\in L^p$ for some $d/(d-1)<p<\infty$. Under these assumptions, equation \eqref{eq:eq} is understood in the sense of distributions. The requirement \eqref{eq:aass} is critical in the $L^r$ scale and there is no reason for $\psi$ to be continuous.

Due to its close connection with the Pauli operator $[\sigma\cdot (-i\nabla -A)]^2 = (-i\nabla -A)^2 -\sigma\cdot B$, equation \eqref{eq:eq} has relevance in various physical contexts.  Zero modes play a role in quantum electrodynamics and in the problem of stability of matter interacting with magnetic fields. We refer to \cite{FrLo1} for further discussion and for references.

Nontrivial solutions $(\psi,A)$ to \eqref{eq:eq} were found in \cite{LoYa}. On the other hand, it is not hard to see, and we shall recall this momentarily, that, if $A$ is small in $L^d$, then \eqref{eq:eq} has only the trivial solution $\psi\equiv 0$. Note that the norm $\|A\|_{L^d}$ is a dimensionless quantity. Our goal here is to find the largest possible upper bound on the $L^d$ norm of $A$ that guarantees the nonexistence of nontrivial solutions. As we shall see, this bound is saturated for the zero modes from \cite{LoYa} and their generalization to higher, odd dimensions in \cite{DuMi}; see also \cite[Appendix]{FrLo1}. Thus, our result characterizes these zero modes as extremizers of an optimization problem. It is of interest that the fields that optimize this variational problem have non-trivial topologies. In fact, the field lines of  the optimizing $A$-field in $d=3$ dimensions are linked circles. The pattern is the one of the Hopf fibration on $\Sph^3$ mapped to $\R^3$ by the stereographic projection.

To appreciate the bound that we will be proving, let us recall the simple argument that shows that, if $A$ is small in $L^d$, then \eqref{eq:eq} has only the trivial solution $\psi\equiv 0$. It is based on the Sobolev inequality
$$
\int_{\R^d} |\nabla u|^2\,dx \geq S_d \left( \int_{\R^d} |u|^{\frac{2d}{d-2}}\,dx \right)^{\frac{d-2}{d}}
\qquad\text{for all}\ u\in\dot H^1(\R^d) \,.
$$
We agree to denote by $S_d$ the \emph{optimal} constant in this inequality. It is known \cite{Rod,Ro,Au,Ta} to have the explicit value
$$
S_d = \frac{d(d-2)}{4}\ |\Sph^d|^{\frac 2d} \,.
$$
If $(\psi,A)$ solves \eqref{eq:eq}, then
$$
\int_{\R^d} |\gamma\cdot(-i\nabla)\psi|^2\,dx = \int_{\R^d} |A|^2|\psi|^2\,dx \,.
$$
We bound the left side from below using the diamagnetic and the Sobolev inequality,
$$
\int_{\R^d} |\gamma\cdot(-i\nabla)\psi|^2\,dx =\int_{\R^d} |\nabla \psi|^2\,dx \geq \int_{\R^d} |\nabla|\psi||^2\,dx \geq S_d \left( \int_{\R^d} |\psi|^{\frac{2d}{d-2}}\,dx \right)^{\frac{d-2}{d}},
$$
and the right side from above using the H\"older inequality,
$$
 \int_{\R^d} |A|^2|\psi|^2\,dx \leq \|A\|_{L^d}^2 \left( \int_{\R^d} |\psi|^{\frac{2d}{d-2}}\,dx \right)^{\frac{d-2}{d}}.
$$
Thus, if $\psi$ is nontrivial, then
$$
\|A\|_{L^d}^2 \geq S_d \,.
$$
Note that through the use of the diamagnetic inequality, i.e., $|\nabla \psi|\ge |\nabla |\psi||$, we destroyed the non-scalar character of the spinor field. For more results on zero modes and their absence, as well as the diamagnetic inequality and its refinements, we refer the reader to the references in \cite{FrLo1}.

Our main result here is that the lower bound on $\|A\|_{L^d}^2$ can be improved to $(d/(d-2))\, S_d$. This is optimal, at least in odd dimensions. Its proof is based on an argument different from \cite{FrLo1}, avoiding the use of any sort of diamagnetic inequality.

Our result is one of the rare instances of a sharp functional inequality for non-scalar objects (vector fields and spinor fields). In contrast, by now there are many results about sharp functional inequalities for scalar objects. Without any attempt at completeness and restricting ourselves to inequalities involving derivatives, we mention as paradigmatic examples the isoperimetric inequality \cite{DG}, Sobolev inequalities \cite{Rod,Ro,Au,Ta}, Hardy--Littlewood--Sobolev inequalities \cite{Lie}, as well as their endpoint cases \cite{Be,CaLo} and some generalizations \cite{JeLe,BrFoMo,FrLi}. In many proofs of these inequalities, rearrangement techniques play an important role. More recently, optimal transport techniques \cite{CENaVi}, flow techniques \cite{CaCaLo,DoEsLo} and reflection techniques \cite{FrLi0} have been successfully employed. As far as we know, none of these techniques has been made to work in a non-scalar setting, and our proof uses different arguments.

Here is the precise statement of our main result.

\begin{theorem}\label{main}
	Let $d\geq 3$. If $\psi\in L^p(\R^d,\C^N)$ for some $d/(d-1)<p<\infty$ is a nontrivial solution of \eqref{eq:eq}, then
	$$
	\|A\|_{L^d}^2 \geq \frac{d}{d-2}\, S_d \,.
	$$
	Equality can be attained if and only if $d$ is odd.	
\end{theorem}

More precisely, in odd dimensions we will characterize all pairs $(\psi,A)$ for which equality in the inequality of the lemma holds. We will state this as Theorem \ref{main2} below.

\begin{remark}\label{gauge}
	Equation \eqref{eq:eq} is gauge-invariant in the sense that if $(\psi,A)$ is a solution of this equation and if $\phi\in L_\loc^1(\R^d,\R)$ is weakly differentiable with $\nabla\phi\in L^d(\R^d,\R^d)$, then $(e^{i\phi}\psi,A+\nabla\phi)$ is also a solution of \eqref{eq:eq} and it satisfies the same integrability assumptions as $(\psi,A)$. Thus, our theorem implies the gauge-invariant bound
	$$
	\inf_\phi \|A-\nabla\phi\|_d^2 \geq \frac{d}{d-2}\, S_d \,.
	$$
	It is not hard to see that there is a unique (up to an additive constant) function $\phi_*$ that minimizes the expression on the left side. Hence, if one sets $A_* := A-\nabla \phi_*$, then using the minimum property one finds that  $\nabla\cdot [|A_*|^{d-2} A_* ] = 0$. One can easily check that the optimizing fields displayed in the next theorem satisfy this equation.
\end{remark}

\begin{remark}
	The problem of minimizing the norm $\|A\|_{L^d}$ among all $A$ that admit non-trivial solutions $\psi$ of \eqref{eq:eq} is conformally invariant, in the sense that, if $\Phi$ is a conformal transformation of $\R^d\cup\{\infty\}$, then $\tilde A(x):= (D\Phi(x))^T A(\Phi(x))$ has the same $L^d$ norm as $A$ and admits a non-trivial solution $\tilde\psi$ of \eqref{eq:eq}. To define $\tilde\psi$, we may use the fact that the conformal group is generated by translations, dilations, orthogonal transformations and inversion and define $\tilde\psi$ only for these generators. For translations and dilations the definition is clear and for orthogonal transformations it appears below in Theorem~\ref{main2}. For the inversion, we define $\tilde\psi(x) := |x|^{-d} \gamma\cdot x \ \psi(x/|x|^2)$ and check that this indeed is a zero mode. Note also that $\tilde\psi$ has the same $L^\frac{2d}{d-1}$-norm as $\psi$.
\end{remark}

\begin{remark}
Inspection of the proof shows that the conclusion of the theorem holds under a somewhat weaker assumption. Namely, if $0\not\equiv\psi \in L^p(\R^d,\C^N)$ with $p= \frac{2d}{d-1}$ satisfies the inequality
$$
|\gamma\cdot \nabla \psi|\le |A||\psi|
\qquad\text{in}\ \R^d \,,
$$
then
$$
	\|A\|_{L^d}^2 \geq \frac{d}{d-2}\, S_d \,.
$$
In this bound, equality can be attained for any (not necessarily odd) $d \ge 3$; see Theorem \ref{main3} in the appendix.
\end{remark}


\subsection*{Characterization of cases of equality}

Throughout this subsection, we assume that $d\geq 3$ is odd. Our goal is to classify all solution pairs $(\psi,A)$ of \eqref{eq:eq} such that $\|A\|_{L^d}^2 = (d/(d-2))S_d$ and $\psi\not\equiv 0$. In essence, our result says that these solution pairs are exactly those constructed in \cite{LoYa} in dimension three, as well as their extension to higher dimensions in \cite{DuMi}. We use the formulation of the latter result in \cite[Appendix]{FrLo1}.

Before stating characterization result, let us review this construction of zero modes. We introduce the $d\times d$ skew symmetric $\Sigma$,
$$
\Sigma := \begin{pmatrix}
	0 & & & & & \\
	& 0 & -1 & & & \\
	& 1 & 0 & & & \\
	& & & \ddots & &\\
	& & & & 0 & -1 \\
	& & & & 1 & 0
\end{pmatrix}.
$$
On the top left corner, there is a zero entry and then there are $(d-1)/2$ blocks of $-i\sigma_2$-matrices on the diagonal. The remaining entries are zero. We define the vector field $\mathcal A:\R^d\to\R^d$ by
$$
\mathcal A(x) := d \left( \frac{1}{1+|x|^2} \right)^2 \left( (1-|x|^2) e_1 + 2 x_1 x + 2\Sigma x\right).
$$
Next, we recall that there is a unique (up to a phase) $\Psi_0\in\C^N$ with $|\Psi_0|=1$ such that
\begin{equation}
	\label{eq:vaccuum}
	\frac12\left( \gamma_{2\alpha} + i\gamma_{2\alpha+1} \right) \Psi_0 =0
	\qquad\text{for all}\ \alpha=1,\ldots,\frac{d-1}{2} \,;
\end{equation}
see \cite[Lemma A.3]{FrLo1} for the existence and \cite[Lemma A.5]{FrLo1} for the uniqueness up to a phase. We know from \cite[Discussion after Lemma A.5]{FrLo1} that there is an $s\in\{+1,-1\}$ such that $\gamma_1 \Psi_0 = s \Psi_0$. We define
$$
\Psi(y) := \left( \frac{1}{1+|x|^2} \right)^\frac{d}{2} \left( 1 +is\gamma\cdot x \right)\Psi_0 \,.
$$
Finally, we recall that for any $O\in\mathcal O(d)$, the orthogonal $d\times d$ matrices, there is a $U\in\mathcal U(N)$, the unitary $N\times N$ matrices, such that
\begin{equation}
	\label{eq:ou}
	U^* \gamma_j U = \sum_{k=1}^d \gamma_k O_{k,j} 
	\qquad\text{for all}\ j=1,\ldots, d\,;
\end{equation}
see \cite[Corollary A.2]{FrLo1}.

A computation (see \cite[Appendix]{FrLo1} and also Section \ref{sec:equality2} below) shows that the pair $(\Psi,\mathcal A)$ solves \eqref{eq:eq} and that $|\mathcal A(x)|= d (1+|x|^2)^{-1}$, so
\begin{equation}
	\label{eq:optconstcomp}
	\|A\|_{L^d}^2 = d^2 \left( \int_{\R^d} \frac{dx}{(1+|x|^2)^d} \right)^{\frac 2d} = d^2 2^{-2} |\Sph^d|^{\frac 2d} = \frac{d}{d-2} \, S_d \,.
\end{equation}
Thus, $(\Psi,\mathcal A)$ saturates the bound in Theorem \ref{main}.

Moreover, for $a\in\R^d$, $b>0,c>0$, $O\in\mathcal O(d)$ and $U\in\mathcal U(N)$, related by \eqref{eq:ou}, the pair
$$
\left( c\, U^*\, \Psi(O^{-1}(x-a)/b) \,,\ b^{-1}\, O\, \mathcal A(O^{-1}(x-a)/b) \right)
$$
is also a solution of \eqref{eq:eq} and the $L^d$ norm of the vector potential is unchanged. Here we use the fact for general spinor fields $\psi$ and $\tilde\psi$ related by $\tilde\psi(x)=U^*\psi(O^{-1}x)$, one has
$$
(\gamma\cdot(-i\nabla)\tilde\psi)(x) = U^* (\gamma\cdot(-i\nabla)\psi)(O^{-1}x) \,.
$$
This follows from a simple computation using \eqref{eq:ou}. Note that besides the parameters $a$, $b$, $c$, $O$ and $U$, there is also a one-dimensional parameter coming from the choice of the phase of $\Psi_0$.

\begin{theorem}\label{main2}
	Let $d\geq 3$ be odd. If $\psi\in L^p(\R^d,\C^N)$ for some $d/(d-1)<p<\infty$ is a nontrivial solution of \eqref{eq:eq} with
	$$
	\|A\|_{L^d}^2 = \frac{d}{d-2}\, S_d \,,
	$$
	then there are $a\in\R^d$, $b>0,c>0$, $O\in\mathcal O(d)$ and $U\in\mathcal U(N)$, related by \eqref{eq:ou}, as well as a $\Psi_0\in\C^N$ with $|\Psi_0|=1$ satisfying \eqref{eq:vaccuum} such that, for all $x\in\R^d$,
	$$
	\psi(x) = c\, U^*\, \Psi(O^{-1}(x-a)/b) 
	\qquad\text{and}\qquad
	A(x) = b^{-1}\, O \, \mathcal A(O^{-1}(x-a)/b) \,.
	$$
\end{theorem}

We emphasize there are solutions to \eqref{eq:eq} different from the extremal ones given in this theorem. In particular, for $\mathcal A$ as above, but multiplied by a certain discrete family of coupling constants $>1$, there are nontrivial solutions to \eqref{eq:eq}; see \cite{LoYa} for $d=3$ and \cite{Mi} for arbitrary odd $d\geq 3$.

\begin{remark}\label{scalarineq}
	In \cite{FrLo1}, besides equation \eqref{eq:eq}, we considered the closely related equation
	\begin{equation}
		\label{eq:eqscalar}
		\gamma\cdot(-i\nabla)\psi = \lambda\,\psi
	\end{equation}
	with a real function $\lambda\in L^d(\R^d)$. We proved that, if $\psi\in L^p(\R^d,\C^N)$ for some $d/(d-1)<p<\infty$ is a nontrivial solution of \eqref{eq:eqscalar}, then
	\begin{equation*}
		\|\lambda\|_{L^d}^2 \geq \frac{d}{d-2}\ S_d \,.
	\end{equation*}
	This inequality is sharp in any, not necessarily odd, dimension $d\geq 3$. The techniques that we develop in the proof of Theorem \ref{main2} allow us to classify the cases of equality in this inequality. We state this as Theorem \ref{main3} in the appendix.
\end{remark}


\subsection*{Relation to Sobolev inequalities} 

In our previous paper \cite{FrLo1} we considered a related, but different problem. There, we were looking for nonexistence results for nontrivial solutions of \eqref{eq:eq} in terms of the norm $\|\nabla\wedge A\|_{L^{d/2}}$. In contrast to our result here, the result in \cite{FrLo1} is probably not optimal. (On the other hand, as mentioned before in Remark \ref{scalarineq} above, \cite{FrLo1} contains an optimal result on a scalar version of this problem.)

In \cite{FrLo1} we also posed the problem of finding the sharp constant $S_d^{\rm v}$ in the Sobolev inequality for vector fields,
\begin{equation}
	\label{eq:sobvector}
	\|\nabla\wedge A\|_{L^{d/2}}^{d/2} \geq S_d^{\rm v}\ \inf_\phi \|A-\nabla\phi\|_d^{d/2} \,.
\end{equation}
In odd dimensions, the vector field $\mathcal A$ satisfies the corresponding Euler--Lagrange equation and it is conceivable that it is an optimizer. If this were true, we could combine the sharp version of \eqref{eq:sobvector} with the inequality in our Theorem \ref{main} here (see also Remark \ref{gauge}) and would obtain an optimal version of the bound in \cite{FrLo1}. Equality would be attained by the same pairs $(\psi,A)$ as given in Theorem \ref{main2}.

We note also that in \cite{FrLo2} we proved both the existence of an optimizer $A$ for \eqref{eq:sobvector} and the existence of optimizing solution pair $(\psi,A)$ such that $\nabla\wedge A$ has minimal $L^{\frac d2}$ norm.

In \cite{FrLo1} we also mentioned a second Sobolev-type inequality, namely, for spinor fields,
\begin{equation}
	\label{eq:sobspinor}
\|\gamma\cdot(-i\nabla)\psi \|_{L^{\frac{2d}{d+1}}}^\frac{2d}{d+1} \geq S_d^{\rm s}\ \|\psi \|_{L^\frac{2d}{d-1}}^\frac{2d}{d+1} \,.
\end{equation}
For any (not necessarily odd) $d\geq 2$, the functions
\begin{equation}
	\label{eq:sobspinoropt}
	\left( \frac{1}{1+|x|^2} \right)^\frac d2 \left( \phi_0 + \gamma\cdot x \phi_1\right)
\end{equation}
with $\phi_0,\phi_1\in\C^N$ with $|\phi_0|=|\phi_1|$ and $\re \langle \phi_0,\gamma_j\phi_1\rangle=0$, $j=1,\ldots,d$, satisfy the corresponding Euler--Lagrange equation and it is conceivable that they are optimizers. If this was true, then the inequality in our main result, Theorem \ref{main}, would immediately follow from
$$
\left\| \gamma\cdot(-i\nabla)\psi \right\|_{L^{\frac{2d}{d+1}}} = \left\| \gamma\cdot A \psi \right\|_{L^{\frac{2d}{d+1}}} = \left\| | A | \psi \right\|_{L^{\frac{2d}{d+1}}} \leq \|A\|_{L^d} \|\psi\|_{L^\frac{2d}{d-1}} \,.
$$
Conversely, our Theorem \ref{main} gives further credence to the conjecture that the sharp constant in \eqref{eq:sobspinor} is attained for the functions in \eqref{eq:sobspinoropt}.

Finding the optimal constants in \eqref{eq:sobvector} and \eqref{eq:sobspinor} remains an \emph{open problem}. 


\subsection*{Idea of the proof}

We emphasize that our proof is valid under the rather weak assumptions $A\in L^d$ and $\psi\in L^p$ for some $d/(d-1)<p<\infty$. In particular, under these assumptions there is no reason for $\psi$ to be continuous. Also, we will need to take derivatives of powers of $|\psi|$, which a priori could lead to problems near the zero set $\{\psi=0\}$. Handling these issues makes our proof somewhat lengthy.

In order to convey the basic idea of our proof, we sketch here the argument ignoring these issues. In other words, we assume that $\psi$ is smooth and non-vanishing. Also, for sake of simplicity, we restrict ourselves to the case where $d=3$.

We start with an integrated version of the Schr\"odinger--Lichnerowicz identity
\begin{align} \label{lichner}
& \sum_{j=1}^3 \int_{\R^3} \left| \left[-i\partial_j - \frac13 \gamma_j \gamma\cdot (-i\nabla) \right](\phi^{\frac 32}\psi) \right|^2 \phi^{-2} \, dx \notag \\
& = \frac{2}{3}  \int_{\R^3} |\gamma \cdot \nabla \psi|^2 \phi \,dx + 2 \int_{\R^3}  |\psi|^2  \frac{\Delta \eta}{\eta} \phi \,dx  \, . 
\end{align}
Here $\psi$ is a smooth spinor, $\phi$ a strictly positive smooth function and 
$$
\eta =\phi^{-\frac{1}{2}} \ . 
$$
The (pointwise) Schr\"odinger--Lichnerowicz identity is named after the papers \cite{Sc,Li}. We apply this pointwise identity on $\R^3$ endowed with $\phi^{-2}$ times the Euclidean metric and with the Dirac and Penrose operators corresponding to this metric. Translating back to the standard metric and integrating we obtain \eqref{lichner}; see \cite[Lemma 3.2 and the discussion afterwards]{HeMo} for a related argument.

Next, in \eqref{lichner} we pick $\phi = |\psi|^{-1}$, i.e.,
$$
\eta = |\psi|^{\frac 12} \,,
$$
and compute 
$$
2 \int_{\R^3}   |\psi|^2  \frac{\Delta \eta}{\eta} \phi \, dx = 2 \int_{\R^3}   |\psi|^{\frac 12} \Delta |\psi|^{\frac12} \,dx = - 2 \int_{\R^3}  \left|\nabla |\psi|^{\frac12}\right|^2 dx \, .
$$
As a consequence of \eqref{lichner}, we find that
$$
 \int_{\R^3} \frac{ |\gamma \cdot \nabla \psi|^2}{|\psi|} dx \ge 3 \int_{\R^3}  \left|\nabla |\psi|^{\frac12}\right|^2 dx  \, .
$$
Applying this inequality to a zero mode $\psi$, i.e., $-i\gamma \cdot \nabla \psi = \gamma \cdot A \psi$ yields
$$
 \int_{\R^3} |A|^2 |\psi| \,dx \ge 3 \int_{\R^3}  \left|\nabla |\psi|^{\frac12}\right|^2 dx \ge 3\,  S_3\, \Vert \psi \Vert_3
$$ 
Applying H\"older's inequality in the left side yields $\Vert A\Vert^2_3 \ge 3S_3$ which is the desired conclusion.

If  $\Vert A\Vert^2_3 = 3S_3$ then there is equality in the Sobolev inequality  and, moreover, the left side of \eqref{lichner}
has to vanish. This means one has to find the twistor spinors, i.e., solutions $\Phi$ of the equations
$$
\left[-i\partial_j - \frac13 \gamma_j \gamma\cdot (-i\nabla)\right] \Phi = 0
\qquad\text{for all}\ j=1,2,3 \,,
$$
which are known.  The cases of equality in the Sobolev inequality are known as well and the relation
$\frac{\psi }{|\psi|^{3/2}} = \varphi$ will yield the optimizing zero modes.

Needless to say that a-priori
we cannot assume that the spinors are smooth, nor do we know that they are nonzero. In the next section
we describe how one can develop a formula like \eqref{lichner} for Sobolev functions.


\subsection*{Acknowledgement}

The authors are grateful to G.\ Carron for making them aware of the paper \cite{HeMo}.


\section{An integral identity}

As mentioned before, the key ingredient in our proof is a certain integral identity. We state the identity for functions in $\dot H^1(\R^d,\C^N)$ (sometimes also denoted by $D^1(\R^d,\C^N)$), which is the space of all weakly differentiable $\psi\in L^1_\loc(\R^d)$ such that $\nabla\psi\in L^2(\R^d)$ and $|\{|\psi|>\tau\}|< \infty$ for all $\tau>0$. Sometimes, for technical reasons, we need to  consider the following regularization of a function $\psi$ on $\R^d$,
$$
\ab := \sqrt{|\psi|^2 + \varepsilon^2} \,,\qquad \varepsilon>0 \,.
$$
This section is devoted to the proving the following result. 

\begin{proposition}\label{identity}
	Let $d\geq 3$. If $\psi\in\dot H^1(\R^d,\C^N)$, then, for all $\epsilon>0$,
	\begin{align*}
		&\int_{\R^d} \sum_{j=1}^d \Big | \left[-i\partial_j - \frac 1d \gamma_j \gamma \cdot (-i\nabla)\right] \frac{\psi}{\ab^{d/(d-1)}}\Big |^2 \ab^2 \,dx\\
		& = \frac{d-1}{d} \int_{\R^d} \frac{|\gamma \cdot \nabla \psi|^2}{\ab^{2/(d-1)}} \,dx 
		- \frac{d-1}{(d-2)^2} \int_{\R^d} \left|\nabla \ab^{\frac{d-2}{d-1}} \right|^2 \left[ 2(d-1) - d\, \frac{ |\psi|^2}{\ab^2} \right] dx \,.
	\end{align*}
\end{proposition}

\begin{proof}
	We will use the short-hand
	$$
	\phi := \frac{\psi}{\ab^{d/(d-1)}} \,.
	$$
	We split the proof into several steps. The starting point of the proof is the following formula, which follows from the properties of the $\gamma$ matrices,
	\begin{equation}
		\label{eq:penroseformula}
		\sum_{j=1}^d \left| \left[-i\partial_j - \frac 1d \gamma_j \gamma \cdot (-i\nabla)\right] \phi \right|^2  = \left|\nabla  \phi \right|^2 - \frac1d \left|\gamma \cdot (-i\nabla) \phi \right|^2.
	\end{equation}
	In the first two steps, we prove pointwise formulas for the two terms on the right side, multiplied by $\ab^2$. In Step 4, which is based on some preparations in Step 3, we will prove an integral formula, which will allow us in Step 5 to conclude the proof of the proposition.	
	
	\medskip
	
	\emph{Step 1.} 
	We claim that
	\begin{align}\label{comp1}
		\left|\nabla \phi \right|^2 \ab^2
		= \frac1{\ab^{2/(d-1)}}\, |\nabla \psi|^2 + \left|\nabla \ab^{\frac{d-2}{d-1}} \right|^2 \left[  \left(\frac{d}{d-2}\right)^2 \frac{|\psi|^2}{\ab^2} - \frac{2d(d-1)}{(d-2)^2}  \right].
	\end{align}
	
	To prove this, we differentiate $\phi$ using the chain rule for weakly differentiable functions as in \cite[Theorem 6.16]{LiLo} and obtain
	\begin{equation}
		\label{eq:derphi}
			\nabla \phi = \ab^{-d/(d-1)} \nabla \psi -\frac{d}{d-1}  \ab^{-(2d-1)/(d-1)}(\nabla \ab) \psi \,,
	\end{equation}
	so that
	\begin{align*}
		\left|\nabla \phi \right|^2\ab^2 & =  \ab^{-\frac{2d}{d-1}+2} |\nabla \psi|^2  + \left(\frac{d}{d-1}\right)^2 \ab^{-\frac{2(2d-1)}{d-1}+4} \left|\nabla \ab \right|^2 \ab^{-2} \left|\psi\right|^2 \\
		& \quad - \frac{2d}{d-1} \ab^{-\frac{3d-1}{d-1}+3}\,\ab^{-1} \nabla \ab \cdot \re \langle \psi, \nabla \psi \rangle \,.
	\end{align*}
	Using $\ab \nabla \ab = |\psi|\nabla |\psi| = \re \langle \psi, \nabla \psi \rangle$, the last two terms simplify to
	\begin{align*}
		& \ab^{-\frac{2}{d-1}} \left| \nabla \ab \right|^2 \left[  \left(\frac{d}{d-1}\right)^2 \frac{|\psi|^2}{\ab^2} 
		- \frac{2d}{d-1}  \right] \\
		& =\left(\frac{d-1}{d-2}\right)^2 \left| \nabla \ab^{\frac{d-2}{d-1}} \right|^2 
		\left[  \left(\frac{d}{d-1}\right)^2 \frac{|\psi|^2}{\ab^2} 
		- \frac{2d}{d-1}  \right] \\
		& = \left| \nabla \ab^{\frac{d-2}{d-1}} \right|^2 \left[  \left(\frac{d}{d-2}\right)^2 \frac{|\psi|^2}{\ab^2} - \frac{2d(d-1)}{(d-2)^2}  \right].
	\end{align*}
	Combining the terms yields \eqref{comp1}, as claimed.
	
\medskip

\emph{Step 2.} 	
	We have that
	\begin{align}\label{comp2}
		\left| \gamma \cdot (-i \nabla ) \phi \right|^2\ab^2 & =
		\frac{1}{\ab^{2/(d-1)}} \left|\gamma \cdot \nabla \psi \right|^2
		+ \left(\frac{d}{d-2}\right)^2  \left|\nabla \ab^{\frac{d-2}{d-1}}\right|^2 \frac{ |\psi|^2}{\ab^2} \notag \\
		& \quad - \frac{2d}{d-1}  \frac{1}{\ab^{2/(d-1)+1}}\, \re \langle \gamma \cdot (\nabla \ab) \psi , \gamma \cdot \nabla \psi \rangle \,.
	\end{align}

	To prove this, we note that \eqref{eq:derphi} implies
	$$
	\gamma \cdot \nabla  \phi =  \ab^{-d/(d-1)} \gamma \cdot \nabla \psi - \frac{d}{d-1}  \ab^{-d/(d-1)-1} \gamma \cdot (\nabla \ab)\, \psi
	$$
	and, using the commutation relations of the $\gamma$ matrices, we find
	\begin{align*}
		\left|\gamma \cdot \nabla \phi \right|^2\ab^2 & = \frac{1}{\ab^{\frac{2d}{d-1}}} \left|\gamma \cdot \nabla \psi \right|^2 \ab^2
		+ \left(\frac{d}{d-1}\right)^2   \frac{1}{\ab^{\frac{2d}{d-1}}} \left|\nabla \ab\right|^2 |\psi|^2 \\
		& \quad -  \frac{2d}{d-1}  \frac{1}{\ab^{\frac{2d}{d-1}-1}} \re \langle \gamma \cdot \nabla \psi, \gamma \cdot (\nabla \ab) \psi\rangle \\
		& = \frac{1}{\ab^{2/(d-1)}} \left|\gamma \cdot \nabla \psi \right|^2
		+ \left(\frac{d}{d-2}\right)^2  \left| \nabla \ab^{\frac{d-2}{d-1}} \right|^2 \frac{ |\psi|^2}{\ab^2} \\
		& -  \frac{2d}{d-1}  \frac{1}{\ab^{2/(d-1)+1}} {\rm Re} \langle \gamma \cdot \nabla \psi, \gamma \cdot (\nabla \ab) \psi\rangle \,.
	\end{align*}
	This proves \eqref{comp2}.
	
\medskip

\emph{Step 3.} We show that, if $\chi \in C^\infty_c(\R^d)$, then
	\begin{align}\label{eq:step3}
		\int_{\R^d}  \ab^{-\frac{2}{d-1}} |\nabla \psi|^2 \chi\, dx & = \frac{2(d-1)}{(d-2)^2} \int_{\R^d} \left|\nabla \ab^{\frac{d-2}{d-1}} \right|^2 \chi \,dx \notag \\
		& \quad	+ \int_{\R^d} \ab^{-\frac{2}{d-1}} \left|\gamma \cdot \nabla \psi \right|^2 \chi\, dx \notag \\
		& \quad - \frac{2}{d-1}\int_{\R^d} \ab^{-1- \frac2{d-1}} \re \langle \gamma \cdot (\nabla \ab) \psi , \gamma\cdot\nabla \psi \rangle \, \chi\,dx \notag \\
		& \quad +\sum_{j, k =1,\, j\not=k}^d \int_{\R^d}  \ab^{-\frac{2}{d-1}} {\rm Re}\langle \gamma_j \psi , \gamma_k \partial_k \psi)\rangle \partial_j \chi \,dx \,.
	\end{align}
	
	To prove this, as a preliminary step, we show that for any bounded, compactly supported function $f$ with $\nabla f\in L^d(\R^d)$ and any $j\neq k$, one has
	\begin{align}\label{eq:step3eq}
		& \int_{\R^d} f \langle \partial_k \psi, \gamma_k \gamma_j \partial_j \psi \rangle \,dx
		+ \int_{\R^d} f \langle \partial_j \psi, \gamma_j \gamma_k \partial_k \psi \rangle \,dx \notag \\
		& = - \int_{\R^d} (\partial_k f) \langle \psi, \gamma_k \gamma_j \partial_j \psi \rangle\, dx - \int_{\R^d} (\partial_j f) \langle \psi, \gamma_j \gamma_k \partial_k \psi \rangle\, dx \,.
	\end{align}
	Since $C^\infty_c(\R^d)$ is dense in $\dot H^1(\R^d)$ (by multiplying by a smooth cut-off function and mollifying), it suffices to prove \eqref{eq:step3eq} for $\psi\in C^\infty_c(\R^d)$. Here we also use that, by Sobolev's inequality, $\psi\in L^{\frac{2d}{d-2}}$, so $(\nabla f)\psi \in L^2$.

	For $\psi\in C^\infty_c(\R^d)$, we integrate by parts in both terms on the left side of \eqref{eq:step3eq} and find
	\begin{align*}
		\int_{\R^d} f \langle \partial_k \psi, \gamma_k \gamma_j \partial_j \psi \rangle \,dx
		& = - \int_{\R^d} (\partial_k f) \langle \psi, \gamma_k \gamma_j \partial_j \psi \rangle \,dx
		- \int_{\R^d} f \langle \psi, \gamma_k \gamma_j \partial_k \partial_j \psi \rangle \,dx \,, \\
		\int_{\R^d} f \langle \partial_j \psi, \gamma_j \gamma_k \partial_k \psi \rangle \,dx
		& = - \int_{\R^d} (\partial_j f) \langle  \psi, \gamma_j \gamma_k \partial_k \psi \rangle \,dx
		- \int_{\R^d} f \langle \psi, \gamma_j \gamma_k \partial_j \partial_k \psi \rangle \,dx \,.
	\end{align*}
	Summing these two equations and using the anticommutation relations to cancel the last term, we obtain \eqref{eq:step3eq}.

	Let us turn to the proof of \eqref{eq:step3}. We may assume that $\phi$ is real-valued. With $f=\chi \ab^{-\frac{2}{d-1}}$, we have
	\begin{align*}
		& \int_{\R^d} \ab^{-\frac{2}{d-1}} \left| \nabla \psi \right|^2 \chi\, dx - \int_{\R^d} \ab^{-\frac{2}{d-1}} \left|\gamma \cdot \nabla \psi \right|^2 \chi\, dx \\
		& = - \sum_{j<k} \int_{\R^d} f \left( \langle \partial_k \psi, \gamma_k \gamma_j \partial_j \psi \rangle + \langle \partial_j \psi, \gamma_j \gamma_k \partial_k \psi \rangle \right) dx \\
		& = \sum_{j<k} \int_{\R^d} \left( (\partial_k f) \langle \psi, \gamma_k \gamma_j \partial_j \psi \rangle + (\partial_j f) \langle \psi, \gamma_j \gamma_k \partial_k \psi \rangle\right) dx \\
		& = \sum_{j,k=1,\, j\not=k}^d  \int_{\R^d} (\partial_j f) \langle \psi, \gamma_j \gamma_k \partial_k \psi \rangle \, dx \,.
	\end{align*}
	We now insert
	$$
	\partial_j f = -\frac{2}{d-1} \ab^{-\frac{2}{d-1}-1} \chi \partial_j \ab + \ab^{-\frac{2}{d-1}}\partial_j \chi
	$$
	(which also implies $\nabla f\in L^d$). After taking the real part, the term involving $\partial_j\chi$ leads to the last term in \eqref{eq:step3}. For the term involving $\partial_j\ab$ we note
	\begin{align*}
		& \sum_{j,k=1,\, j\not=k}^d  \int_{\R^d} \ab^{-\frac{2}{d-1}-1}\partial_j\ab \re \langle \psi, \gamma_j \gamma_k \partial_k \psi \rangle \, \chi\,dx \\
		& = \int_{\R^d} \! \ab^{-\frac{2}{d-1}-1} \! \re \langle \gamma\cdot(\nabla\ab) \psi, \gamma\cdot\nabla \psi \rangle \, \chi\,dx - \int_{\R^d} \! \ab^{-\frac{2}{d-1}-1} \! \left( \nabla \ab\right) \cdot \re \langle \psi, \nabla \psi \rangle \, \chi\,dx
	\end{align*}
	and, using again $\ab\nabla\ab = \re\langle\psi,\nabla\psi\rangle$, we write
	$$
	\ab^{-\frac{2}{d-1}-1} \left( \nabla \ab\right) \cdot \re \langle \psi, \nabla \psi \rangle 
	= \ab^{-\frac{2}{d-1}} |\nabla\ab|^2 = \left( \frac{d-1}{d-2}\right)^2 \left| \nabla\ab^\frac{d-2}{d-1} \right|^2.
	$$
	In this way, we arrive at \eqref{eq:step3}.
	
\medskip

\emph{Step 4.}
	We claim that
	\begin{align}\label{comp3}
		\int_{\R^d}  \ab^{-\frac{2}{d-1}} \left|\nabla \psi\right|^2 dx & = \frac{2(d-1)}{(d-2)^2} \int_{\R^d} \left|\nabla \ab^{\frac{d-2}{d-1}} \right|^2 dx \notag \\
		& \quad + \int_{\R^d} \ab^{-\frac{2}{d-1}} \left|\gamma \cdot \nabla \psi \right|^2 dx \notag \\
		& \quad - \frac{2}{d-1} \int_{\R^d} \ab^{-1-2/(d-1)}  \re \langle \gamma \cdot (\nabla\ab) \psi, \gamma\cdot\nabla \psi \rangle \,dx \,.
	\end{align}	
	
	Choose $\chi\in C^\infty_c(\R^d)$ be equal to one near the origin and apply the equality in Step 3 with $\chi_R(x) := \Phi(x/R)$. Since $\nabla\psi\in L^2$ and $|\psi|_\epsilon\geq \epsilon$ we have
	\begin{align*}
		& \lim_{R\to\infty} \int_{\R^d}  \ab^{-\frac{2}{d-1}} |\nabla \psi|^2 \chi_R \, dx = \int_{\R^d}  \ab^{-\frac{2}{d-1}} |\nabla \psi|^2\, dx \,, \\
		& \lim_{R\to\infty} \int_{\R^d} \ab^{-\frac{2}{d-1}} \left|\gamma \cdot \nabla \psi \right|^2 \chi_R \, dx = \int_{\R^d} \ab^{-\frac{2}{d-1}} \left|\gamma \cdot \nabla \psi \right|^2 \, dx \,, \\
		& \lim_{R\to\infty} \int_{\R^d} \ab^{-1-\frac2{d-1}} \re \langle \gamma \cdot (\nabla\ab) \psi , \gamma\cdot\nabla \psi \rangle \chi_R \,dx \\
		& \qquad = \int_{\R^d} \ab^{-1-\frac2{d-1}}  \re \langle  \gamma \cdot (\nabla\ab) \psi, \gamma\cdot\nabla \psi \rangle \,dx \,.
	\end{align*}
	Moreover, if $\Phi$ is chosen radially nonincreasing, then, by monotone convergence,
	$$
	\lim_{R\to\infty} \int_{\R^d} \left|\nabla \ab^{\frac{d-2}{d-1}} \right|^2 \chi_R \,dx = \int_{\R^d} \left|\nabla \ab^{\frac{d-2}{d-1}} \right|^2 dx \,.
	$$
	Thus, to complete the proof, we need to show that for $j\neq k$,
	$$
	\lim_{R\to\infty} \int_{\R^d}  \ab^{-\frac{2}{d-1}} {\rm Re}\langle \gamma_j \psi , \gamma_k \partial_k \psi)\rangle \partial_j \chi_R \,dx = 0 \,.
	$$
	To prove this, we bound $|\partial_j\chi_R|\leq \const |x|^{-1}\1_{\{|x|\geq cR\}}$, where $\Phi\equiv 1$ on $\{|x|\leq c\}$. By Hardy's inequality, $|x|^{-1}\psi\in L^2$. This, together with $\nabla\psi\in L^2$ and $\ab\geq\epsilon$, implies the claimed limit by dominated convergence. This completes the proof of \eqref{comp3}.

\medskip

\emph{Step 5.} We now conclude the proof of the proposition.
	Inserting \eqref{comp1} and \eqref{comp2} into \eqref{eq:penroseformula}, we obtain
	\begin{align*}
		& \sum_{j=1}^d \left| \left[-i\partial_j - \frac 1d \gamma_j \gamma \cdot (-i\nabla)\right] \phi \right|^2 \ab^2 \\
		& = \frac1{\ab^{2/(d-1)}} |\nabla \psi|^2 + \left|\nabla \ab^{\frac{d-2}{d-1}} \right|^2 \left[  \left(\frac{d}{d-2}\right)^2 \frac{|\psi|^2}{\ab^2} - \frac{2d(d-1)}{(d-2)^2}  \right] \\
		& \quad - \frac1d \frac{1}{\ab^{2/(d-1)}} \left|\gamma \cdot \nabla \psi \right|^2 
		- \frac1d \left(\frac{d}{d-2}\right)^2  \left|\nabla \ab^{\frac{d-2}{d-1}}\right|^2 \frac{ |\psi|^2}{\ab^2} \notag \\
		& \quad + \frac{2}{d-1}  \frac{1}{\ab^{2/(d-1)+1}}\, \re \langle \gamma \cdot (\nabla \ab) \psi , \gamma \cdot \nabla \psi \rangle \,.
	\end{align*}
	We integrate this formula over $\R^d$ and use \eqref{comp3} to express the integral of
	the first and last term on the right side in terms of integrals involving $|\nabla \ab^{\frac{d-2}{d-1}} |^2$ and $|\gamma\cdot\nabla\psi|^2$. Collecting terms, we arrive at the claimed identity in the proposition.
\end{proof}


\section{Proof of the inequality}

In this short section, we deduce Theorem \ref{main} from Proposition \ref{identity}. Let $(\psi,A)$ be a solution of \eqref{eq:eq} satisfying $\psi\in L^p(\R^d,\C^N)$ for some $d/(d-1)<p<\infty$ and $A\in L^d(\R^d,\R^d)$. Then, as shown in \cite{FrLo1}, $\psi\in L^{\frac{2d}{d-2}}(\R^d,\C^N)$. Since $A\in L^d(\R^d,\R^d)$, we deduce from H\"older's inequality that $\gamma\cdot A \psi\in L^2(\R^d,\R^d)$. Thus, by \eqref{eq:eq}, $\gamma\cdot(-i\nabla)\psi \in L^2(\R^d)$ and, consequently, $\psi\in\dot H^1(\R^d)$. Therefore, we can apply Proposition \ref{identity}. Dropping the nonnegative term on the left side and using $|\psi|\leq\ab$ on the right side, we obtain
$$
\frac{d-1}{d-2} \int_{\R^d} \left| \nabla \ab^{\frac{d-2}{d-1}}\right|^2\,dx \leq \frac{d-1}{d} \int_{\R^d} \frac{|\gamma\cdot\nabla\psi|^2}{\ab^{2/(d-1)}}\,dx = \frac{d-1}{d} \int_{\R^d} \frac{|A|^2 |\psi|^2}{\ab^{2/(d-1)}}\,dx \,.
$$
We bound the left side from below with Sobolev's inequality,
$$
\int_{\R^d} \left| \nabla \ab^{\frac{d-2}{d-1}}\right|^2\,dx
\geq S_d \left( \int_{\R^d} \left( \ab^{\frac{d-2}{d-1}} - \epsilon^{\frac{d-2}{d-1}} \right)^{\frac{2d}{d-2}}dx \right)^{\frac{d-2}{d}},
$$
and the right side from above with H\"older's inequality,
$$
\int_{\R^d} \frac{|A|^2 |\psi|^2}{\ab^{2/(d-1)}}\,dx \leq \int_{\R^d} |A|^2 |\psi|^{\frac{2(d-2)}{d-1}} \,dx
\leq \|A\|_{L^d}^2 \left( \int_{\R^d} |\psi|^{\frac{2d}{d-1}}\,dx \right)^{\frac{d-2}{d}}.
$$
Thus, we obtain
$$
\frac{S_d}{d-2} \left( \int_{\R^d} \left( \ab^{\frac{d-2}{d-1}} - \epsilon^{\frac{d-2}{d-1}} \right)^{\frac{2d}{d-2}}dx \right)^{\frac{d-2}{d}}
\leq  \frac{\|A\|_{L^d}^2}{d} \left( \int_{\R^d} |\psi|^{\frac{2d}{d-1}}\,dx \right)^{\frac{d-2}{d}}.
$$
We now let $\epsilon\to 0$. Since $\epsilon\mapsto \ab^{\frac{d-2}{d-1}} - \epsilon^{\frac{d-2}{d-1}}$ is pointwise nonincreasing, we can use monotone convergence and obtain
$$
\frac{S_d}{d-2} \left( \int_{\R^d} |\psi|^{\frac{2d}{d-1}}\,dx \right)^{\frac{d-2}{d}}
\leq  \frac{\|A\|_{L^d}^2}{d} \left( \int_{\R^d} |\psi|^{\frac{2d}{d-1}}\,dx \right)^{\frac{d-2}{d}}.
$$
Since $\psi\not\equiv 0$, we obtain the claimed lower bound on $\|A\|_d^2$. This concludes the proof.


\section{Characterizing cases of equality. I}\label{sec:equality1}

We now investigate the cases of equality in the bound in Theorem \ref{main}. In this section, as a first step, we discuss the absolute value of $\psi$ and $A$. We shall prove the following result.

\begin{proposition}\label{equality1}
	Let $d\geq 3$. If $\psi\in L^p(\R^d,\C^N)$ for some $d/(d-1)<p<\infty$ is a nontrivial solution of \eqref{eq:eq} with
	$$
	\|A\|_{L^d}^2 = \frac{d}{d-2}\, S_d \,,
	$$
	then there are $a\in\R^d$, $b>0,c>0$ such that, for all $x\in\R^d$,
	$$
	|\psi(x)| = c \left( \frac{b^2}{b^2+|x-a|^2} \right)^\frac{d-1}{2}
	\qquad\text{and}\qquad
	|A(x)| = d\ \frac{b}{b^2+|x-a|^2} \,.
	$$
	Moreover,
	\begin{equation}
		\label{eq:equalitytwistor}
		\left[-i\partial_j - \frac 1d \gamma_j \gamma \cdot (-i\nabla)\right] \frac{\psi}{|\psi|^{d/(d-1)}} \equiv 0
		\qquad\text{for all}\ j=1,\ldots, d	\,.
	\end{equation}
\end{proposition}

We prove this proposition by rewriting the proof in the previous section, keeping track of all the nonnegative terms that we dropped in that argument.

\begin{proof}
	Let us abbreviate
$$
P_\epsilon := \int_{\R^d} \sum_{j=1}^d \Big | \left[-i\partial_j - \frac 1d \gamma_j \gamma \cdot (-i\nabla)\right] \frac{\psi}{\ab^{d/(d-1)}}\Big |^2 \ab^2 \,dx
$$
and
$$
R_\epsilon := 
\frac{d(d-1)}{(d-2)^2}  \int_{\R^d} |\nabla \ab^{\frac{d-2}{d-1}}|^2 \frac{ \epsilon^2}{\ab^2}\, dx \,.
$$
Then the identity in Proposition \ref{identity} can be written as
$$
R_\epsilon + P_\epsilon = \frac{d-1}{d} \int_{\R^d} \frac{|\gamma \cdot \nabla \psi|^2}{\ab^{2/(d-1)}}dx - \frac{d-1}{d-2}
\int_{\R^d} \left|\nabla \ab^{\frac{d-2}{d-1}} \right|^2 \, dx \,.
$$
From equation \eqref{eq:eq}, we get
$$
R_\epsilon + P_\epsilon =  \frac{d-1}{d} \int_{\R^d} \frac{|A|^2 |\psi|^2}{\ab^{2/(d-1)}}dx - \frac{d-1}{d-2}
\int_{\R^d} \left|\nabla \ab^{\frac{d-2}{d-1}} \right|^2 \, dx \,.
$$
We want to apply the H\"older and Sobolev inequality to the two terms on the right side, respectively. We therefore write
$$
R_\epsilon + P_\epsilon + R_\epsilon^{(1)} + R_\epsilon^{(2)} = S_\epsilon \,,
$$
where
$$
R_\epsilon^{(1)} := \frac{d-1}{d} \left( \|A\|_{L^d}^2 \left( \int_{\R^d} |\psi|^{\frac{2d}{d-2}} \ab^{-\frac{2d}{(d-1)(d-2)}}\,dx \right)^{\frac{d-2}{d}} - \int_{\R^d} \frac{|A|^2 |\psi|^2}{\ab^{2/(d-1)}}dx \right),
$$
$$
R_\epsilon^{(2)}:= \frac{d-1}{d-2} \left( 
\int_{\R^d} \left|\nabla \ab^{\frac{d-2}{d-1}} \right|^2 \, dx
- S_d \left( \int_{\R^d} \left( \ab^{\frac{d-2}{d-1}} - \epsilon^{\frac{d-2}{d-1}} \right)^{\frac{2d}{d-2}} dx \right)^{\frac{d-2}{d}} \right)
$$
and
\begin{align*}
	S_\epsilon & := \frac{d-1}{d} \|A\|_{L^d}^2 \left( \int_{\R^d} |\psi|^{\frac{2d}{d-2}} \ab^{-\frac{2d}{(d-1)(d-2)}} \,dx \right)^{\frac{d-2}{d}} \\
	& \quad - \frac{d-1}{d-2} S_d \left( \int_{\R^d} \left( \ab^{\frac{d-2}{d-1}} - \epsilon^{\frac{d-2}{d-1}} \right)^{\frac{2d}{d-2}} dx \right)^{\frac{d-2}{d}}.
\end{align*}
By monotone convergence, together with the fact that $\psi\in L^{\frac{2d}{d-1}}$, it is easy to see that
$$
\lim_{\epsilon\to 0} S_\epsilon = \left( \frac{d-1}{d} \|A\|_{L^d}^2 - \frac{d-1}{d-2}\, S_d \right) \left( \int_{\R^d} |\psi|^{\frac{2d}{d-1}} \, dx \right)^{\frac{d-2}{d}}.
$$
On the other hand, since each one of the terms $R_\epsilon$, $P_\epsilon$, $R_\epsilon^{(1)}$ and $R_\epsilon^{(2)}$ is nonnegative, we have $S_\epsilon\geq 0$. Since $\psi\not\equiv 0$, we conclude again that
$$
\| A\|_{L^d}^2 \geq \frac{d}{d-2}\, S_d \,,
$$
which is the bound we derived in the previous subsection.

\medskip

Now assume that
$$
\| A\|_{L^d}^2 = \frac{d}{d-2}\, S_d \,.
$$
Then, by the above argument, $\lim_{\epsilon\to 0} S_\epsilon=0$ and, consequently,
\begin{equation}
	\label{eq:equalitylim}
	\lim_{\epsilon\to 0} R_\epsilon = \lim_{\epsilon\to 0} P_\epsilon = \lim_{\epsilon\to 0} R_\epsilon^{(1)} = \lim_{\epsilon\to 0} R_\epsilon^{(2)} = 0 \,.
\end{equation}
(The existence of these four limits is part of the conclusion.) 

\medskip

Let us begin with the term $R_\epsilon^{(1)}$. Using monotone convergence, we find that
$$
\lim_{\epsilon\to 0} R_\epsilon^{(1)} = \frac{d-1}{d} \left(  \|A\|_{L^d}^2 \left( \int_{\R^d} |\psi|^{\frac{2d}{d-1}} \,dx \right)^{\frac{d-2}{d}} - \int_{\R^d} |A|^2 |\psi|^{\frac{2(d-2)}{d-1}} \, dx \right).
$$
Thus, from \eqref{eq:equalitylim} we conclude that
$$
 \|A\|_{L^d}^2 \left( \int_{\R^d} |\psi|^{\frac{2d}{d-1}} \,dx \right)^{\frac{d-2}{d}} = \int_{\R^d} |A|^2 |\psi|^{\frac{2(d-2)}{d-1}} \, dx
$$
and, therefore, by the characterization of equality in H\"older's inequality,
\begin{equation}
	\label{eq:equalityaabs}
	|A| = \const |\psi|^{\frac{2}{d-1}}
\end{equation}
for some positive constant.

\medskip

Next, we consider $R_\epsilon^{(2)}$. We note that $\ab^{\frac{d-2}{d-1}} - \epsilon^{\frac{d-2}{d-1}}$ converges pointwise monotonically to $|\psi|^{\frac{d-2}{d-1}}$ as $\epsilon\to 0$. By monotone convergence,
$$
\lim_{\epsilon\to 0} \int_{\R^d} \left( \ab^{\frac{d-2}{d-1}} - \epsilon^{\frac{d-2}{d-1}} \right)^{\frac{2d}{d-2}} dx = \int_{\R^d} |\psi|^{\frac{2d}{d-1}} \,dx \,.
$$
This, together with the fact that $R_\epsilon^{(2)}$ tends to zero (by \eqref{eq:equalitylim}) and therefore, in particular, remains bounded, implies that $\int_{\R^d} |\nabla \ab^{\frac{d-2}{d-1}}|^2 \, dx$ remains bounded. Moreover, by monotone convergence, $\ab^{\frac{d-2}{d-1}}\to |\psi|^{\frac{d-2}{d-1}}$ in $L^1_\loc(\R^d)$. By a simple argument (see Lemma \ref{weakdiff} below), these facts imply that $|\psi|^{\frac{d-2}{d-1}}$ is weakly differentiable in $\R^d$ and
$$
\int_{\R^d} \left|\nabla |\psi|^{\frac{d-2}{d-1}}\right|^2\,dx \leq \liminf_{\epsilon\to 0} \int_{\R^d} \left| \nabla \ab^{\frac{d-2}{d-1}} \right|^2 \, dx \,.
$$
We conclude that
$$
\liminf_{\epsilon\to 0} R_\epsilon^{(2)} \geq \frac{d-1}{d-2} \left( 
\int_{\R^d} |\nabla |\psi|^{\frac{d-2}{d-1}}|^2 \, dx
- S_d \left( \int_{\R^d} |\psi|^{\frac{2d}{d-1}} \,dx \right)^{\frac{d-2}{d}} \right).
$$
By \eqref{eq:equalitylim} and Sobolev's inequality we conclude that
$$
\int_{\R^d} \left|\nabla |\psi|^{\frac{d-2}{d-1}} \right|^2 \, dx
= S_d \left( \int_{\R^d} |\psi|^{\frac{2d}{d-1}} \,dx \right)^{\frac{d-2}{d}}.
$$
By the characterization of cases of equality in Sobolev's inequality (see, e.g., \cite[Theorem 8.3]{LiLo} for a textbook presentation), we have, for some $a\in\R^d$ and $b,c>0$,
$$
|\psi(x)|^{\frac{d-2}{d-1}} = c^{\frac{d-2}{d-1}} \left( \frac{b^2}{b^2 + |x-a|^2} \right)^{\frac{d-2}2}.
$$
This proves the form of $|\psi|$ stated in the proposition.

We draw one more conclusion, which we will not use, but which might be useful in another context. Namely, since the lower semicontinuity inequality for the weak convergence is saturated, the weak convergence is, in fact, strong convergence, that is,
$$
\nabla \ab^{\frac{d-2}{d-1}} \to \nabla |\psi|^{\frac{d-2}{d-1}}
\qquad\text{in}\ L^2(\R^d) \,.
$$

\medskip

Returning with the form of $|\psi|$ to \eqref{eq:equalityaabs}, we find that
$$
|A(x)| = \const \frac{b^2}{b^2+ |x-a|^2}
$$
with some positive constant. This constant can be determined in view of the computation in \eqref{eq:optconstcomp} and the assumption that $\|A\|_{L^2}^2 = (d/(d-2))S_d$. This yields the form of $|A|$ stated in the proposition.

\medskip

Finally, we consider the term $P_\epsilon$. Since we have already shown that $|\psi|$ is locally bounded away from zero, it is easy to see that $\psi \ab^{-d/(d-1)}\to \psi |\psi|^{-d/(d-1)}$ in $L^1_\loc(\Omega)$. Therefore, as in the lemma, the distribution
$$
\left[-i\partial_j - \frac 1d \gamma_j \gamma \cdot (-i\nabla)\right] \frac{\psi}{|\psi|^{d/(d-1)}}
$$
is an $L^2$ function and
\begin{align*}
	& \int_{\R^d} \Big | \left[-i\partial_j - \frac 1d \gamma_j \gamma \cdot (-i\nabla)\right] \frac{\psi}{|\psi|^{d/(d-1)}}\Big |^2 |\psi|^2 dx \\
	& \leq\liminf_{\epsilon\to 0} \int_{\R^d} \Big | \left[-i\partial_j - \frac 1d \gamma_j \gamma \cdot (-i\nabla)\right] \frac{\psi}{\ab^{d/(d-1)}}\Big |^2 \ab^2 dx \,.
\end{align*}
Thus,
$$
\int_{\R^d}\sum_{j=1}^d  \Big | \left[-i\partial_j - \frac 1d \gamma_j \gamma \cdot (-i\nabla)\right] \frac{\psi}{|\psi|^{d/(d-1)}}\Big |^2 |\psi|^2 dx \leq \liminf_{\epsilon\to 0} P_\epsilon \,.
$$
By \eqref{eq:equalitylim}, we conclude that
$$
\int_{\R^d}\sum_{j=1}^d  \Big | \left[-i\partial_j - \frac 1d \gamma_j \gamma \cdot (-i\nabla)\right] \frac{\psi}{|\psi|^{d/(d-1)}}\Big |^2 |\psi|^2 dx = 0 
$$
and, consequently, recalling also that $|\psi|\neq 0$, we obtain equation \eqref{eq:equalitytwistor}. This completes the proof of the proposition.
\end{proof}

In the previous proof, we used the following simple lemma.

\begin{lemma}\label{weakdiff}
	Let $\Omega\subset\R^d$ be open, let $f_n\in L^1_\loc(\Omega)$ be weakly differentiable in $\Omega$ and $f_n\to f$ in $L^1_\loc(\Omega)$. Assume that $(\nabla f_n)$ is bounded in $L^p(\Omega)$ for some $1<p<\infty$. Then $f$ is weakly differentiable in $\Omega$ and $(\nabla f_n)$ converges weakly to the weak gradient of $f$. In particular,
	$$
	\int_\Omega |\nabla f|^p\,dx \leq \liminf_{n\to\infty} \int_\Omega |\nabla f_n|^p\,dx \,.
	$$
\end{lemma}

\begin{proof}
	Let $F$ be a weak limit point of $(\nabla f_n)$ in $L^p(\Omega)$. Such a weak limit point exists by weak compactness. Then, with limits taken along the corresponding subsequence, for any $\phi\in C^1_c(\Omega)$,
	$$
	\int_\Omega f\partial_k\phi\,dx = \lim_{n\to\infty} \int_\Omega f_n \partial_k\phi\,dx = - \lim_{n\to\infty} \int_\Omega \partial_k f_n \phi\,dx = \int_\Omega F_k \phi\,dx \,.
	$$
	This shows that $f$ is weakly differentiable with $\nabla f= F$. Since the weak gradient is unique, there is a unique weak limit point of $(\nabla f_n)$, so, $(\nabla f_n)$ converges weakly.
\end{proof}


\section{Twistor spinors}

In the previous section, we determined the absolute values of  $\psi$ and $A$ of extremal solutions of the inequality in Theorem \ref{main}. As a step towards determining the `argument' $\psi/|\psi|$, in this section, we will characterize all solutions of equation \eqref{eq:equalitytwistor}.

\begin{theorem}\label{twistor}
	Let $d\geq 3$ and assume that $\Phi$ is a spinor field on $\R^d$ satisfying
	\begin{equation}
		\label{eq:twistoreq}
		\left[-i\partial_j - \frac 1d \gamma_j \gamma \cdot (-i\nabla)\right] \Phi =0
		\qquad\text{for all}\ j=1,\ldots,d \,.
	\end{equation}
	Then there are constant spinors $\phi_0,\phi_1\in\C^N$ such that
	$$
	\Phi(x) = \phi_0 + \gamma\cdot x\, \phi_1 
	\qquad\text{for all}\ x\in\R^d \,.
	$$
\end{theorem}

This theorem is known. It appears, for instance, in \cite{Fr}. The equation for $\Phi$ is called the twistor equation and its solutions are called twistor spinors. We include the proof of the theorem for the sake of concreteness and since it simplifies considerably in the present Euclidean context.

\begin{proof}
	A priori, we only assume that $\Phi$ is a distribution that satisfies the equation in distributional sense. Then mollifications of $\Phi$ are smooth functions which satisfy the same twistor equation. Assuming the theorem has been proved for smooth functions, we conclude that each mollification has the form in the theorem with constant spinors $\phi_0$ and $\phi_1$ which depend on the mollification parameter. Since the mollifications converge to $\Phi$ in the sense of distributions as the mollification parameter vanishes, it is easy to see that the parameters $\phi_0$ and $\phi_1$ converge and, consequently, $\Phi$ has the claimed form.
	
	Thus, from now on, we may assume that $\Phi$ is a smooth function. (In fact, $C^2$ is enough.) We differentiate the equation in \eqref{eq:twistoreq} with respect to $x_k$ and obtain, abbreviating $D:=\gamma\cdot(-i\nabla)$,
	\begin{equation}
		\label{eq:twistoreqdiff}
		\partial_k\partial_j \Phi = \frac{i}{d} \gamma_j \partial_k D\Phi 
		\qquad\text{for all}\ j,k=1,\ldots,d \,.
	\end{equation}
	Taking $k=j$ and summing, we obtain
	$$
	\Delta \Phi = - \frac{i}{d} D^2\Phi \,.
	$$
	Since $D^2=-\Delta$, we conclude that
	\begin{equation}
		\label{eq:harmonic}
		D^2\Phi =0 \,.
	\end{equation}
	Next, we use \eqref{eq:twistoreqdiff} twice to get
	$$
	\gamma_j \partial_k D\Phi = -id \partial_k \partial_j \Phi = -id\partial_j\partial_k\Phi = \gamma_k \partial_j D\Phi \,.
	$$
	Multiplying by $\gamma_j$ and using the anticommutation relations, we deduce
	$$
	\partial_k D\Phi = - \gamma_k \gamma_j \partial_j D\Phi + 2\delta_{j,k} \partial_k D\Phi \,.
	$$
	Summing with respect to $j$ gives
	$$
	d \partial_k D\Phi = -i \gamma_k D^2\Phi + 2 \partial_k D\Phi 
	$$
	The assumption $d\geq 3$ and \eqref{eq:harmonic} imply that
	$$
	\partial_k D\Phi = 0 
	\qquad\text{for all}\ k=1,\ldots,d \,.
	$$
	This implies that there is a $\phi_1\in\C^N$ such that
	$$
	D\Phi = \phi_1 \,.
	$$
	Inserting this information into \eqref{eq:twistoreq} gives
	$$
	-i\partial_j \Phi - \frac{1}{d} \gamma_j \phi_1 = 0
	\qquad\text{for all}\ j=1,\ldots, d \,.
	$$
	Thus, there is a $\phi_0\in\C^N$ such that
	$$
	\Phi(x) = \phi_0 + \frac{i}{d} \gamma\cdot x\,\phi_1 \,.
	$$
	This is the assertion, up to redefining $\phi_1$.
\end{proof}


\section{Characterizing cases of equality. II}\label{sec:equality2}

Our goal in this section is to complete the proof of Theorem \ref{main2} concerning the characterization of extremal solutions of the inequality in Theorem \ref{main}. As a byproduct, we will also prove the claim in Theorem \ref{main} that the inequality there is not attained in even dimensions.

We assume throughout this section that $(\psi,A)$ solves \eqref{eq:eq}, that $0\not\equiv \psi\in L^p(\R^d,\C^N)$ for some $d/(d-1)<p<\infty$ and that $\|A\|_{L^2}^d = (d/(d-2))S_d$.

According to Proposition \ref{equality1} and after translating and dilating $\psi$ and $A$ and multiplying $\psi$ by a constant, we may, without loss of generality, assume that
\begin{equation}
	\label{eq:equalityabs}
	|\psi(x)| = \left( \frac{1}{1+|x|^2} \right)^\frac{d-1}{2}
	\qquad\text{and}\qquad
	|A(x)| = d\ \frac{1}{1+|x|^2} \,.
\end{equation}


\subsection*{Using the twistor equation}

According to Proposition \ref{equality1}, $\psi/|\psi|^{d/(d-1)}$ satisfies the twistor equation \eqref{eq:twistoreq}. Thus, by Theorem \ref{twistor}, there are $\phi_0,\phi_1\in\C^N$ such that
$$
\frac{\psi(x)}{|\psi(x)|^{d/(d-1)}} = \phi_0 + \gamma\cdot x\,\phi_1 
\qquad\text{for all}\ x\in\R^d \,.
$$
Taking absolute values in the latter equation gives
\begin{align*}
	|\psi(x)|^{-\frac1{d-1}} & = |\phi_0 + \gamma\cdot x\, \phi_1| = \left( |\phi_0|^2 + 2\re \langle \phi_0, \gamma\cdot x\,\phi_1 \rangle + |\gamma\cdot x\,\phi_1|^2 \right)^{\frac 12} \\
	& = \left( |\phi_0|^2 + 2\re \langle \phi_0, \gamma\cdot x\,\phi_1\rangle + |\phi_1|^2 |x|^2 \right)^{\frac 12}
\end{align*}
Comparing this with the formula for $|\psi|$ in \eqref{eq:equalityabs} gives
$$
1+ |x|^2 = |\phi_0|^2 + 2\re \langle \phi_0, \gamma\cdot x\,\phi_1 \rangle + |\phi_1|^2 |x|^2 \,,
$$
that is,
\begin{equation}
	\label{eq:parameters}
	|\phi_1|=|\phi_0| =1 \,,
	\qquad
	\re \langle \phi_0, \gamma_j \phi_1 \rangle = 0
	\qquad\text{for all}\ j=1,\ldots,d \,.
\end{equation}
To summarize, we know at the moment that
\begin{equation}
	\label{eq:charpsi}
	\psi(x) = \left( \frac{1}{1+|x|^2} \right)^{\frac d2} (\phi_0 + \gamma\cdot x\,\phi_1)
\end{equation}
with $\phi_0,\phi_1$ satisfying \eqref{eq:parameters}.


\subsection*{Recovering the vector potential}

For $1\leq j,k\leq d$ we compute, using the properties of the $\gamma$ matrices,
$$
\re \langle \psi, \gamma_j\gamma_k\psi \rangle = \frac12\left( \langle \psi, \gamma_j\gamma_k\psi \rangle + \langle \psi, \gamma_k\gamma_j\psi\rangle \right) = \delta_{j,k} |\psi|^2 \,.
$$
Thus,
$$
\re \langle \psi,\gamma_j \gamma\cdot A\,\psi \rangle = \sum_k A_k  \re \langle \psi, \gamma_j\gamma_k\psi \rangle = A_j |\psi|^2 \,.
$$
On the other hand, by \eqref{eq:eq},
$$
\re \langle \psi,\gamma_j \gamma\cdot A\,\psi \rangle = \re \langle \psi,\gamma_j \gamma\cdot(-i\nabla) \psi \rangle
$$
and, therefore,
$$
A_j = \frac{\re \langle \psi,\gamma_j \gamma\cdot(-i\nabla) \psi \rangle}{|\psi|^2} \,.
$$
Using \eqref{eq:charpsi}, we compute
$$
\partial_k \psi = -d  \left( \frac{1}{1+|x|^2} \right)^{\frac{d+2}2} x_k (\phi_0+\gamma\cdot x\phi_1) +  \left( \frac{1}{1+|x|^2} \right)^{\frac d2} \gamma_k\phi_1
$$
and
\begin{align}\label{eq:equalitydirac}
	\gamma\cdot(-i\nabla)\psi(x) & = id  \left( \frac{1}{1+|x|^2} \right)^{\frac{d+2}2} \gamma\cdot x(\phi_0+\gamma\cdot x\phi_1) -i d \left( \frac{1}{1+|x|^2} \right)^{\frac d2} \phi_1 \notag \\
	& = -i d  \left( \frac{1}{1+|x|^2} \right)^{\frac{d+2}2}  (\phi_1 - \gamma\cdot x \phi_0) \,.
\end{align}
Inserting this into the above formula for $A_j$, we find
\begin{align*}
	A_j & = d  \left( \frac{1}{1+|x|^2} \right)^2 \im \langle (\phi_0 + \gamma\cdot x \phi_1), \gamma_j (\phi_1-\gamma\cdot x\phi_0) \rangle \\
	& = d  \left( \frac{1}{1+|x|^2} \right)^2 \im \left( \langle\phi_0,\gamma_j\phi_1 \rangle - \langle\phi_0,\gamma_j \gamma\cdot x\phi_0\rangle \right. \\
	& \qquad\qquad\qquad\qquad\quad \left. + \langle\phi_1,\gamma\cdot x \gamma_j \phi_1 \rangle - \langle \phi_1,\gamma\cdot x \gamma_j \gamma\cdot x \phi_0 \rangle \right).
\end{align*}
This expression can be slightly simplified with the help of the following lemma.

\begin{lemma}\label{commutation}
	For all $x,y\in\R^d$,
	$$
	\gamma\cdot x\ \gamma\cdot y\ \gamma\cdot x = - |x|^2 \gamma\cdot y + 2(x\cdot y) \gamma\cdot x \,.
	$$
\end{lemma}

\begin{proof}
	By linearity, we may assume that $x=e_j$. We write
	\begin{align*}
		\gamma\cdot x \gamma_j \gamma\cdot x = \sum_{k,\ell} x_k x_\ell \gamma_k\gamma_j\gamma_\ell = \sum_k x_k^2 \gamma_k \gamma_j \gamma_k + \sum_{k<\ell} x_k x_\ell \left( \gamma_k \gamma_j\gamma_\ell + \gamma_\ell \gamma_j \gamma_k \right).
	\end{align*}
	By the anticommutation relations, $\gamma_k\gamma_j\gamma_k = - \gamma_j + 2\delta_{k,j} \gamma_k$, so
	$$
	\sum_k x_k^2 \gamma_k \gamma_j \gamma_k = - |x|^2 \gamma_j + 2 x_j^2 \gamma_j \,.
	$$
	Similarly, if $k<\ell$, then $\gamma_k\gamma_j\gamma_\ell +\gamma_\ell\gamma_j\gamma_k = - \gamma_j (\gamma_k\gamma_\ell + \gamma_\ell\gamma_k) + 2\delta_{k,j}\gamma_\ell + 2 \delta_{\ell,j} \gamma_k= 2(\delta_{k,j}\gamma_\ell + \delta_{\ell,j}\gamma_k)$ and so
	$$
	\sum_{k<\ell} x_k x_\ell \left( \gamma_k \gamma_j\gamma_\ell + \gamma_\ell \gamma_j \gamma_k \right) = 2 \sum_{k<\ell} x_k x_\ell \left( \delta_{k,j}\gamma_\ell + \delta_{\ell,j}\gamma_k \right) = 2 \sum_{k\neq j} x_k x_j \gamma_k \,.
	$$
	This proves the claimed formula.
\end{proof}

Inserting the formula from the lemma into the previous equation for $A_j$ gives
\begin{align*}
	A_j = d  \left( \frac{1}{1+|x|^2} \right)^2 & \im \left((1-|x|^2) \langle\phi_0,\gamma_j\phi_1 \rangle + 2 x_j \langle \phi_0, \gamma\cdot x \phi_1 \rangle \right. \\
	& \qquad \left. - \langle\phi_0,\gamma_j \gamma\cdot x\phi_0\rangle - \langle\phi_1,\gamma_j \gamma\cdot x \phi_1 \rangle \right).
\end{align*}
Introducing the vector $w\in\R^d$ by
$$
w_j:= \im\langle\phi_0,\gamma_j\phi_1\rangle
$$
as well as the matrix $M\in\R^{d\times d}$ by
$$
M_{j,k} := - \frac12( \im \langle\phi_0,\gamma_j \gamma_k \phi_0\rangle + \im \langle\phi_1,\gamma_j \gamma_k \phi_1 \rangle) \,,
$$
we can write this as
\begin{equation}
	\label{eq:equalitya}
	A = d  \left( \frac{1}{1+|x|^2} \right)^2 \left((1-|x|^2) w + 2 x (w\cdot x) + 2 Mx \right).
\end{equation}
Note that, by \eqref{eq:parameters}, we have $w_j = -i \langle\phi_0,\gamma_j\phi_1\rangle$. Also, by the anticommutation relations, we see that the matrix $M$ is skew-symmetric, that is,
$$
M^T = - M \,.
$$

\bigskip

Next, we derive equations for the matrix $M$ and the vector $w$. They imply, in particular, that $d$ is odd. Since $|A(x)|= d/(1+|x|^2)$, we must have
$$
1+|x|^2 = | (1-|x|^2) w + 2x (w\cdot x) + 2 Mx |
\qquad\text{for all}\ x\in\R^d \,.
$$
Since
\begin{align*}
	& | (1-|x|^2) w + 2x (w\cdot x) + 2 Mx |^2 \\
	& = (1-|x|^2)^2 |w|^2 + 4 |Mx|^2 + 4 (w\cdot x)^2 + 4 (1-|x|^2) (w\cdot Mx) +8 (x\cdot Mx) (w\cdot x) \,.
\end{align*}
By skew-symmetry, we have $x\cdot Mx =0$. Since the right side above is equal to $(1+|x|^2)^2$, the odd-degree term $(1-|x|^2) (w\cdot Mx)$ must vanish, that is, by skew-symmetry,
$$
Mw = 0 \,.
$$
The remaining equations are
$$
|w|^2 = 1 \,,\qquad
|x|^2 = - |w|^2 |x|^2 + 2 |Mx|^2 + 2 (w\cdot x)^2 \,.
$$
In view of the first equation here, the second one is equivalent to
$$
M^T M + |w\rangle\langle w| = 1 \,.
$$
This implies, in particular, that $\ker M={\rm span}\{ w\}$. Since the dimension of the kernel of a skew-symmetric matrix in even dimension is even dimensional, we conclude that \emph{$d$ is odd}.


\subsection*{Using the zero mode equation}

In what follows we assume that $d$ is odd. It follows from \eqref{eq:charpsi} and \eqref{eq:equalitya} that
$$
\gamma\cdot A \psi = d \left( \frac{1}{1+|x|^2} \right)^{\frac{d+4}2} \left( (1-|x|^2) \gamma\cdot w + 2(w\cdot x) \gamma\cdot x + 2 \gamma\cdot Mx \right) (\phi_0 + \gamma\cdot x\phi_1) \,.
$$
Combining this with \eqref{eq:equalitydirac} and the equation \eqref{eq:eq}, we get
$$
-i (1+|x|^2)  (\phi_1 - \gamma\cdot x \phi_0) =  \left( (1-|x|^2) \gamma\cdot w + 2(w\cdot x) \gamma\cdot x + 2 \gamma\cdot Mx \right) (\phi_0 + \gamma\cdot x\phi_1) \,.
$$
Using Lemma \ref{commutation}, we can rewrite this as
$$
-i (1+|x|^2)  (\phi_1 - \gamma\cdot x \phi_0) =  \left( \gamma\cdot w + \gamma\cdot x\ \gamma\cdot w\ \gamma\cdot x + 2 \gamma\cdot Mx \right) (\phi_0 + \gamma\cdot x\phi_1) \,.
$$
Both sides are polynomials of degree three. For us, only the equation that is obtained for homogeneity one is interesting, namely
\begin{align}
	\label{eq:equality2}
	i\gamma\cdot x\ \phi_0 & = \gamma\cdot w\ \gamma\cdot x\ \phi_1 + 2\gamma\cdot Mx\ \phi_0
	\qquad\text{for all}\ x\in\R^d \,.
\end{align}
(In fact, one can show that this equation is equivalent to the one corresponding to homogeneity two and that those corresponding to homogeneities zero and three are consequences of the above equation.) From \eqref{eq:equality2} we derive
\begin{align}
	\label{eq:equality1}
	-i\phi_1 & = \gamma\cdot w\ \phi_0 \,,\\
	\label{eq:equality21}
	\gamma\cdot My\ \phi_0 & = i\gamma\cdot y\ \phi_0
	\qquad\text{for all}\ y\in w^\bot \,.
\end{align}
Indeed, \eqref{eq:equality1} follows by taking $x=w$ in \eqref{eq:equality2} and recalling that $|w|=1$ and $Mw=0$. Let us prove \eqref{eq:equality21}. It follows from the properties of the gamma matrices that
\begin{equation*}
	\gamma\cdot w\ \gamma\cdot x= - \gamma\cdot x\ \gamma\cdot w + 2(w\cdot x) \,.
\end{equation*}
Inserting this into \eqref{eq:equality2} and using \eqref{eq:equality1}, we obtain
$$
2i \gamma\cdot x\ \phi_0 = 2(w\cdot x)\phi_1 + 2 \gamma\cdot Mx\ \phi_0
\qquad\text{for all}\ x\in\R^d \,.
$$
Specializing to $x$ orthogonal to $w$ yields \eqref{eq:equality21}.

\bigskip

After these preparations we are in position to complete the proof of our second main result.

\begin{proof}[Proof of Theorem \ref{main2}]
Recall the definition of the matrix $\Sigma$ before Theorem \ref{main2}. Since $M$ is skew-symmetric and satisfies $M^TM+|w\rangle\langle w|=1$, there is an $O\in \mathcal O(d)$ such that
$$
O^T M O = \Sigma \,.
$$
We note that $MOe_1 = O \Sigma e_1 =0$. Since $M^TM+|w\rangle\langle w|=1$, this implies that $Oe_1=w$. Thus, we can rewrite \eqref{eq:equalitya} as
$$
A(x) = d \left( \frac{1}{1+|x|^2} \right)^2 \left( (1-|x|^2) O e_1 + 2x(e_1\cdot O^{-1}x) + O\Sigma O^{-1} x \right) = O \mathcal A(O^{-1}x) \,.
$$
Thus, $A$ is of the form claimed in the theorem.

Next, given the matrix $O\in\mathcal O(d)$, there is a $U\in\mathcal U(N)$ such that \eqref{eq:ou} holds; see \cite[Corollary A.2]{FrLo1}. We now show that $U\phi_0$ is a vaccuum, that is, it satisfies
\begin{equation}
	\label{eq:vaccuum2}
	\frac12\left( \gamma_{2\alpha} + i \gamma_{2\alpha+1}\right) U\phi_0 = 0
	\qquad\text{for all}\ \alpha = 1,\ldots,\frac{d-1}{2} \,.
\end{equation}
Indeed, since $\Sigma e_{2\alpha+1} = -e_{2\alpha}$, we have
$$
U^* \gamma_{2\alpha} U = \gamma\cdot O e_{2\alpha} = - \gamma\cdot O\Sigma e_{2\alpha+1} = - \gamma\cdot MOe_{2\alpha+1} \,,
$$
so, using \eqref{eq:equality21},
$$
U^*\gamma_{2\alpha} U\phi_0 = -\gamma\cdot MO e_{2\alpha+1}\phi_0 = -i \gamma\cdot Oe_{2\alpha+1}\phi_0 = -i U^* \gamma_{2\alpha+1} U\phi_0 \,.
$$
This proves \eqref{eq:vaccuum2}.

Next, we note that
$$
\frac12\left( \gamma_{2\alpha} + i \gamma_{2\alpha+1}\right) \gamma_1 U\phi_0 = 0
\qquad\text{for all}\ \alpha = 1,\ldots,\frac{d-1}{2} \,.
$$
Indeed, this follows immediately from \eqref{eq:vaccuum2}, since $\gamma_1$ anticommutes with $\gamma_{2\alpha}$ and $\gamma_{2\alpha+1}$ for $\alpha\geq 1$.

Thus, we have shown that both $U\phi_0$ and $\gamma_1 U\phi_0$ are vaccua. By the uniqueness of the vaccuum \cite[Lemma A.5]{FrLo1}, there is a $\lambda\in\C$ such that $\gamma_1U\phi_0 = \lambda U\phi_0$. Since $|\gamma_1 U\phi_0|=|U \phi_0|$, we have $|\lambda|=1$ and, since $\gamma_1$ is Hermitian, we have $\lambda\in\R$. Thus, $s:=\lambda\in\{+1,-1\}$ and $\gamma_1 U\phi_0 = s U\phi_0$.

The equality $Oe_1=w$ implies $U^*\gamma_1 U = \gamma\cdot O e_1 = \gamma\cdot w$. Thus, by \eqref{eq:equality1}, $U^*\gamma_1 U\phi_0 = \gamma\cdot w\phi_0 = -i\phi_1$. We conclude that
$$
\phi_1 = i U^*\gamma U\phi_0 = i s\phi_0 \,.
$$

Since $|U\phi_0|=|\phi_0|=1$, by uniqueness of the vaccuum (see \cite[Lemma A.5]{FrLo1}) we may assume that $U\phi_0 = \Psi_0$. Note that above, we showed that $\gamma_1 \Psi_0 = \gamma_1 U\phi_0 = s U\phi_0 = s\Psi_0$, which justifies the notation $s$. Moreover, we can rewrite \eqref{eq:charpsi} as
$$
\psi(x) = \left( \frac{1}{1+|x|^2} \right)^{\frac d2} (1+is\gamma\cdot x)\phi_0 = \left( \frac{1}{1+|x|^2} \right)^{\frac d2} U^* (1+is U \gamma\cdot x U^*) \Psi_0 = U^*\Psi(O^{-1}x) \,.
$$
Here in the last equality we used \eqref{eq:ou}. Thus, $\psi$ is of the form claimed in the theorem. This completes the proof.
\end{proof}


\appendix

\section{Characterizing cases of equality in another inequality}

In this appendix, we consider the equation
\begin{equation}
	\label{eq:eqscalarapp}
	\gamma\cdot(-i\nabla)\psi = \lambda\,\psi
\end{equation}
with a real function $\lambda\in L^d(\R^d)$. In \cite{FrLo1}, we proved that, if $\psi\in L^p(\R^d,\C^N)$ for some $d/(d-1)<p<\infty$ is a nontrivial solution of \eqref{eq:eqscalarapp}, then
\begin{equation*}
	\|\lambda\|_{L^d}^2 \geq \frac{d}{d-2}\ S_d \,.
\end{equation*}
(Note that in \cite{FrLo1} we used a slightly different normalization.) A simply computation shows that equality is attained for the pair $(\tilde\Psi,\Lambda)$, where
$$
\tilde\Psi(x) := \left( \frac{1}{1+|x|^2} \right)^\frac d2 \left( 1 + is\gamma\cdot x\right) \phi_0 \,,
\qquad
\Lambda(x) := sd\ \frac{1}{1+|x|^2} \,.
$$
Here $\phi_0\in\C^N$ is a constant spinor and $s\in\{+1,-1\}$. Note that, in contrast to the situation of Theorem \ref{main2}, the constant spinor $\phi_0$ is \emph{not} required to satisfy the vaccuum conditions \eqref{eq:vaccuum} and $s$ is not coupled to $\phi_0$. The following theorem shows that, up to translations, dilations and multiplications by constants, this family constitutes the only pairs for which equality is attained.

\begin{theorem}\label{main3}
	Let $d\geq 3$. If $\psi\in L^p(\R^d,\C^N)$ for some $d/(d-1)<p<\infty$ is a nontrivial solution of \eqref{eq:eqscalarapp} with
	$$
	\|\lambda\|_{L^d}^2 = \frac{d}{d-2}\ S_d \,,
	$$
	then there are $a\in\R^d$, $b>0,c>0$, as well as a $\phi_0\in\C^N$ with $|\phi_0|=1$ and an $s\in\{+1,-1\}$ such that, for all $x\in\R^d$,
	$$
	\psi(x) = c\, \tilde\Psi((x-a)/b) 
	\qquad\text{and}\qquad
	\lambda(x) = b^{-1}\, \Lambda((x-a)/b) \,.
	$$
\end{theorem}

\begin{proof}
	We argue as in the proof of Theorem \ref{main2}. In the same way as in Proposition \ref{equality1} we deduce that, after translating and dilating $(\psi,\lambda)$ and multiplying $\psi$ by a constant,
	$$
	|\psi(x)| = \left( \frac{1}{1+|x|^2} \right)^\frac{d-1}{2}
	\qquad\text{and}\qquad
	|\lambda(x)| = d\ \frac{1}{1+|x|^2} \,.
	$$
	Moreover, we obtain equations \eqref{eq:twistoreqdiff}, which, according to Theorem \ref{twistor}, implies the form \eqref{eq:charpsi} of $\psi$ with $\phi_0,\phi_1\in\C^N$ satisfying \eqref{eq:parameters}. Thus $\gamma\cdot(-i\nabla)\psi$ is given by \eqref{eq:equalitydirac}, and inserting this into \eqref{eq:eqscalarapp}, we find
	\begin{equation}
		\label{eq:equalityapp}
		-i\left( \phi_1 -\gamma\cdot x\ \phi_0 \right) = \frac{\lambda(x)}{|\lambda(x)|} \left( \phi_0 + \gamma\cdot x\ \phi_1 \right)
		\qquad\text{for all}\ x\in\R^d \,.
	\end{equation}
	Taking the real part of the inner product of this equation with $\phi_0$ and recalling \eqref{eq:parameters}, we find that
	$$
	s:= \im\langle\phi_0,\phi_1 \rangle = \frac{\lambda(x)}{|\lambda(x)|} 
	\qquad\text{for all} x\in\R^d \,.
	$$
	This shows that the sign of $\lambda$ is constant. Returning with this information to \eqref{eq:equalityapp} and evaluating at $x=0$, we infer that $-i\phi_1= s\phi_0$. This leads to the claimed form of $\psi$ and $\lambda$ and completes the proof.
\end{proof}


\bibliographystyle{amsalpha}

\begin{thebibliography}{26}

\bibitem{Au} T.~Aubin, \textit{Probl\`emes isoperim\'etriques et espaces de Sobolev}. J. Differ. Geometry \textbf{11} (1976), 573--598.

\bibitem{Be} W.\ Beckner, \textit{Sharp Sobolev inequalities on the sphere and the Moser--Trudinger inequality}. Ann.\ of Math.\ (2) \textbf{138} (1993), no.\ 1, 213--242.

\bibitem{BrFoMo} T.\ P.\ Branson, L.\ Fontana, C.\ Morpurgo, \textit{Moser--Trudinger and Beckner--Onofri's inequalities on the CR sphere}. Ann.\ of Math.\ (2) \textbf{177} (2013), no.\ 1, 1--52. 

\bibitem{CaLo} E.\ Carlen, M.\ Loss, \textit{Competing symmetries, the logarithmic HLS inequality and Onofri's inequality on $\Sph^n$}. Geom.\ Funct.\ Anal.\ \textbf{2} (1992), no.\ 1, 90--104.

\bibitem{CaCaLo} E.\ A.\ Carlen, J.\ A.\ Carrillo, M. Loss, \textit{Hardy--Littlewood--Sobolev inequalities via fast diffusion flows}. Proc.\ Natl.\ Acad.\ Sci.\ USA \textbf{107} (2010), no.\ 46, 19696--19701.

\bibitem{CENaVi} D.\ Cordero-Erausquin, B.\ Nazaret, C.\ Villani, \textit{A mass-transportation approach to sharp Sobolev and Gagliardo-Nirenberg inequalities}. Adv.\ Math.\ \textbf{182} (2004), no.\ 2, 307--332. 

\bibitem{DG} E.\ De Giorgi, \textit{Sulla propriet\`a isoperimetrica dell’ipersfera, nella classe degli insiemi aventi frontiera orientata di misura finita}. Atti Accad.\ Naz.\ Lincei Mem.\ Cl.\ Sci.\ Fis.\ Mat.\ Nat., Sez.\ I, \textbf{8} (1958), 33--44.

\bibitem{DoEsLo} J.\ Dolbeault, M.\ J.\ Esteban, M.\ Loss, \textit{Rigidity versus symmetry breaking via nonlinear flows on cylinders and Euclidean spaces}. Invent.\ Math.\ \textbf{206} (2016), no.\ 2, 397--440.

\bibitem{DuMi}
Gerald~V. Dunne and Hyunsoo Min.
\newblock Abelian zero modes in odd dimensions.
\newblock {\em Phys. Rev. D}, 78(6):067701, 4, 2008.	

\bibitem{Fe} P.\ M.\ N.\ Feehan, \textit{A Kato--Yau inequality and decay estimate for eigenspinors}. J.\ Geom.\ Anal.\ \textbf{11} (2001), no.\ 3, 469--489.

\bibitem{FrLi0} R.\ L.\ Frank, E.\ H.\ Lieb, \textit{Inversion positivity and the sharp Hardy-Littlewood-Sobolev inequality}. Calc.\ Var.\ Partial Differential Equations \textbf{39} (2010), no.\ 1-2, 85--99. 

\bibitem{FrLi} R.\ L.\ Frank, E.\ H.\ Lieb, \textit{Sharp constants in several inequalities on the Heisenberg group}. Ann.\ of Math.\ (2) \textbf{176} (2012), no.\ 1, 349--381.

\bibitem{FrLo1} R.\ L.\ Frank, M.\ Loss, \textit{Which magnetic fields support a zero mode?}. J.\ Reine Ang.\ Math., to appear. arXiv:2012.13646.

\bibitem{FrLo2} R.\ L.\ Frank, M.\ Loss, \textit{Existence of optimizers in a Sobolev inequality for vector fields}. Preprint (2021), arXiv:2107.06450.
	
\bibitem{Fr} T.\ Friedrich, \textit{On the conformal relation between twistors and Killing spinors}. Proceedings of the Winter School on Geometry and Physics ({S}rn\'{\i}, 1989). Rend.\ Circ.\ Mat.\ Palermo (2) Suppl.\ No.\ \textbf{22} (1990), 59--75.	

\bibitem{HeMo} M.\ Herzlich, A.\ Moroianu, \textit{Generalized Killing spinors and conformal eigenvalue estimates for Spinc manifolds}. Ann.\ Global Anal.\ Geom.\ \textbf{17} (1999), no.\ 4, 341--370. 

\bibitem{JeLe} D.\ Jerison, J.\ M.\ Lee, \textit{Extremals for the Sobolev inequality on the Heisenberg group and the CR Yamabe problem}. J.\ Amer.\ Math.\ Soc.\ \textbf{1} (1988), no. 1, 1--13.

\bibitem{LiLo} E.\ H.\ Lieb, M.\ Loss, \textit{Analysis}. Second edition. Graduate Studies in Mathematics, \textbf{14}. American Mathematical Society, Providence, RI, 2001.

\bibitem{Li} A.\ Lichnerowicz, \textit{Spineurs harmoniques}. C.\ R.\ Acad.\ Sci.\ Paris \textbf{257} (1963), 7--9.

\bibitem{Lie} E.\ H.\ Lieb, \textit{Sharp constants in the Hardy--Littlewood--Sobolev and related inequalities}. Ann.\ of Math.\ (2) \textbf{118} (1983), no.\ 2, 349--374.

\bibitem{LoYa} M.\ Loss, H.-T.\ Yau, \textit{Stability of Coulomb systems with magnetic fields. III. Zero energy bound states of the Pauli operator}. Comm.\ Math.\ Phys.\ \textbf{104} (1986), no.\ 2, 283--290. 

\bibitem{Mi} H.\ Min, \textit{Fermion zero modes in odd dimensions}. J.\ Phys.\ A: Math.\ Theor.\ \textbf{43} (2010), 095402.

\bibitem{Rod} E. Rodemich, \textit{The Sobolev inequality with best possible constant}. Analysis Seminar Caltech, Spring 1966.

\bibitem{Ro} G. Rosen, \textit{Minimum value for $c$ in the Sobolev inequality $\|\phi^3\|\leq c\|\nabla\phi\|^3$}. SIAM J. Appl. Math. \textbf{21} (1971), 30--32.

\bibitem{Sc} E.\ Schr\"odinger, \textit{Diracsches Elektron im Schwerefeld}. Sitzungsber.\ Preu\ss.\ Akad.\ Wiss.\, Phys.-Math.\ Kl.\ (1932), no. 11-12, 105--128.

\bibitem{Ta} G. Talenti, \textit{Best constants in Sobolev inequality}. Ann. Mat. Pura Appl. \textbf{110} (1976), 353--372.	
	
\end{thebibliography}

\end{document}